\documentclass[a4paper,12pt]{article}

\usepackage[T1]{fontenc}
\usepackage{lmodern}
\usepackage[english]{babel}
\usepackage{amsmath,amsfonts,amssymb,amsthm,bbm}
\usepackage{units}
\usepackage{graphicx}
\usepackage{color}
\usepackage[nointegrals]{wasysym}

\newtheorem{lemma}{Lemma}[section]
\newtheorem{prop}[lemma]{Proposition}
\newtheorem{thm}[lemma]{Theorem}
\newtheorem{remark}[lemma]{Remark}
\newtheorem{example}[lemma]{Example}
\newtheorem{defin}[lemma]{Definition}

\renewenvironment{proof}[1][]{\par \textbf{Proof #1.}}{\par\medskip\hfill\hbox{$\quad\Box$}\par\medskip\ignorespacesafterend}
\newenvironment{coroll}[1][]{\par \textbf{Corollary #1.}}{\par\medskip\ignorespacesafterend}

\newcommand{\N}{\mathbb{N}}

\newcommand{\Z}{\mathbb{Z}}

\newcommand{\MEGA}{\vcenter{\hbox{\includegraphics[width=5mm]{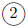}}}}
\newcommand{\MEGISTON}{\vcenter{\hbox{\includegraphics[width=6mm]{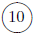}}}}

\title{Beyond Knuth's notation for ``Unimaginable Numbers'' within computational number theory}
\author{Antonino Leonardis - Gianfranco D'Atri - Fabio Caldarola}
\begin{document}

\begin{center}
{\def\baselinestretch{1.5}
\Large \bf Beyond Knuth's notation for ``Unimaginable\\
\vspace{2mm}
Numbers'' within computational number theory}
\vspace{6mm}

{\large\bf Antonino Leonardis$^1$ - Gianfranco d'Atri$^2$ - Fabio Caldarola$^3$}
\vspace{10mm}

{\scriptsize
$^1$ Department of Mathematics and Computer Science, University of Calabria \\ 
Arcavacata di Rende, Italy \\
e-mail: {antonino.leonardis@unical.it}
\vspace{4mm}

$^2$ Department of Mathematics and Computer Science, University of Calabria \\ 
Arcavacata di Rende, Italy
\vspace{4mm}

$^3$ Department of Mathematics and Computer Science, University of Calabria \\ 
Arcavacata di Rende, Italy
e-mail: {caldarola@mat.unical.it}
\vspace{4mm}
}
\end{center}
\vspace{10mm}
\section*{Abstract}
Literature considers under the name \emph{unimaginable numbers} any positive integer going beyond any physical application, with this being more of a vague description of what we are talking about rather than an actual mathematical definition (it is indeed used in many sources without a proper definition). This simply means that research in this topic must always consider shortened representations, usually involving \emph{recursion}, to even being able to describe such numbers.\par\medskip
One of the most known methodologies to conceive such numbers is using \emph{hyper-operations}, that is a sequence of binary functions defined recursively starting from the usual chain: addition - multiplication - exponentiation. The most important notations to represent such hyper-operations have been considered by Knuth, Goodstein, Ackermann and Conway as described in this work's introduction.\par\medskip
Within this work we will give an axiomatic setup for this topic, and then try to find on one hand other ways to represent unimaginable numbers, as well as on the other hand applications to computer science, where the algorithmic nature of representations and the increased computation capabilities of computers give the perfect field to develop further the topic, exploring some possibilities to effectively operate with such big numbers.\par\medskip
After the introduction, we will give axioms and generalizations for the up-arrow notation.\par\smallskip
In the subsequent section we consider a representation via rooted trees of the \emph{hereditary base-$n$ notation} which can be used efficiently to represent some defective unimaginable numbers. This notation is used in the formulation of \emph{Goodstein's theorem} (see \cite{G2}) asserting that the so called ``Goodstein sequences'' eventually terminate at zero, and we will develop this topic by determining in some cases an explicit recursive algorithm for the number of steps required to reach zero, as well as an effective bound for it using Knuth's notation.\par\smallskip
In the last section we will analyse some methods to compare big numbers, proving specifically a theorem about approximation using scientific notation and a theorem on hyperoperation bounds for Steinhaus-Moser notation.
\section{Introduction}
Several methods and notations are been developed in the last century to work, or better to try to consider, very large numbers for which in this paper we propose the name of \emph{unimaginable numbers}.
One of the most known methodologies is the so-called \emph{Knuth up-arrow notation} introduced by D.E. Knuth in 1976 (see \cite{K}) and strictly linked to the concept of \emph{hyper-operation} and \emph{Ackermann function} (see \cite{A}, \cite{N}).\\
The idea of hyper-operation dates back to the early 1900s by A.A. Bennet (see \cite{B}), and subsequently we refind it in a group of Hilbert's students as  W. Ackermann and G. Sudan. But the widespread contemporary names like \emph{tetration}, \emph{pentation}, \emph{hexation}, or in general \emph{hyper-$n$} operation were introduced by R.L. Goodstein in 1947 (see \cite{G1}) and gained popularity through Rudy Rucker's book \emph{Infinity and the Mind} \cite{R}, published in 1982. 
Knuth up-arrow is not the only notation used today for very large numbers; there are in fact many other ways to write hyper-operators, as we may recall among others: 
\begin{itemize}
	\item \emph{square bracket notation}, \emph{box notation} and \emph{superscripts and subscripts notation} (see \cite{Mu1} and \cite{Mu2});
	\item \emph{Nambiar's notation} (see \cite{N}).
\end{itemize}
Moreover we point out that there are also so enormous numbers that even Knuth's notation and the previous ones, are not sufficient to represent them. For this purpose J.H. Conway introduced a more powerful notation based on recursivity, to write extremely large numbers. It is known as \emph{Conway's chained arrow notation} (see for example \cite{C}) and can be viewed as a generalization of Knuth's arrow notation: in fact, in the case of a lenght 2 sequence $a\to b\to n$, it is equivalent to $a\uparrow^n b$ Knuth's notation. Similarly, the \emph{Bowers' operator}, also called the \emph{Bowers' exploding array function} (see \cite{Bow}), is a more powerful numeral system proposed by J. Bowers and published on the web in 2002, which generalizes hyper-operators.\par\medskip
The Steinhaus-Moser notation (see \cite{SM}) is another way to express by recursion very big numbers. It is in fact more intuitive (thus fitting well within educational purposes) in its definition than the hyper-operations, and for its recursion properties will be applied to find an effective bound for certain couples in Goodstein's theorem (see below).\par\medskip
A relevant link between unimaginable numbers and computer science is related with the so called arbitrary-precision arithmetic and blockchain tools, as one can use such huge numbers to handle machine-computed big data. This work arose indeed from a discussion between the authors (during preparation of ``The First Symposium of the International Pythagorean School -- da Pitagora a Sch\"utzenberger: numeri inimmaginabil\^i\^i\^i''\footnote{``Inimmaginabili'' is the italian plural word for ``unimaginable'', and has been modified by using the fancy letter ``\^i'' in order to resemble Knuth's up-arrows.}) about the use of \emph{gross-one}, a recent definition of an arithmetical infinity (see \cite{Y}, \cite{Cal1}, \cite{Cal2} and the
references therein), in order to compute limits in a similar fashion to non-standard analysis; this ``infinite number'' has the flaw of having still a slightly poor axiomatic definition behind it so that in most applications it becomes more convenient to just consider a very big number, in fact an ``unimaginable'' one (more precisely its factorial so that all ``imaginable'' numbers are its divisors, so to respect one of gross-one's fundamental properties).\par\medskip
We will start the paper by giving a complete axiomatic definition of hyper-operators, linking this to Knuth's and Goodstein's notations. We will define the notion of \emph{meta-algorithm} in order to define precisely the idea behind ``repeating'' an operation. After that, we will define a graph-theory representation of numbers linked to Goodstein's theorem (see \cite{G2}), which has also a simple set-theory interpretation when considering base $2$, called \emph{rooted tree representation}, and we will determine in some cases an explicit recursive algorithm for the number of steps required to reach zero for the so called ``Goodstein sequences'', as well as an effective bound for this number using Knuth's notation. We will conclude this work by applying various methods, among others from continued fractions (see \cite{CS}), to compare unimaginable numberss.
\section{Extending Knuth's up-arrow notation}
\subsection{Historical notes}
The basic arithmetical operations are defined recursively starting from the successor operation. The exponentiation, for instance, is a repeated multiplication. Knuth and Goodstein (see \cite{K} and \cite{G1}) have further extended this definition, so that for example the tetration is a repeated exponentiation.
\subsection{Arrow function definition}
The work from Knuth and Goodstein can be formalized by the following general arrow-function:
\begin{enumerate}
	\item $\uparrow(A,B,0):=AB$;
	\item $\uparrow(A,0,k):=1$ for $k\geq 1$;
	\item\label{item:rec.law} $\uparrow(A,B+1,k):=\uparrow(A,\uparrow(A,B,k),k-1)$.
\end{enumerate}
We can add in the mix also the following cases (satisfying recurrence law \ref{item:rec.law}. as well):
\begin{itemize}
	\item $\uparrow(A,B,-2):=A\bowtie B:=\max(A,B)+1$;
	\item $\uparrow(A,B,-1):=A+B$;
	\item $\uparrow(A,-1,k):=0$ for $k\geq 2$.
\end{itemize}
This is a slightly modified version of the original one from Goodstein, which is related by the simple equality:
$$G(k,A,B)=\uparrow(A,B,k-2)$$
and Knuth's notation is as well very similar, writing:
\begin{align*}
A\uparrow B&:=\uparrow(A,B,1)&\text{\bfseries[Normal exponentiation]};\\
A\uparrow\uparrow B&:=\uparrow(A,B,2)&\text{\bfseries[Tetration]};\\
A\uparrow\uparrow\uparrow B&:=\uparrow(A,B,3)&\text{\bfseries[Pentation]};\\
A\uparrow^k B&:=\uparrow(A,B,k)&\text{\bfseries[$k$-th hyper-operation]}.
\end{align*}
The last one is a compact expression for $A\uparrow\ldots\uparrow B$ where $A$ and $B$ are separated by exactly $k$ arrows.\par\medskip
One could also use the symbol $\hat{\ }$ instead of each up-arrow, reobtaining the usual notation for exponentiation.\par\medskip
\textbf{Important remark}: after the normal multiplication, all the operations we have defined are \textbf{no more} commutative nor associative, and priority is to compute them all in order from right to left (\emph{right associativity}).
\begin{example}
Let's compute the following \emph{tetration}:
$$3\uparrow\uparrow 4=3\uparrow 3\uparrow 3\uparrow 3=3\uparrow 3\uparrow 27=3\uparrow 7625597484987$$
which is a number with exactly 3638334640025 digits.
\end{example}
\begin{example}
Let's compute the following \emph{pentation}:
\begin{align*}
2\uparrow\uparrow\uparrow 3&=2\uparrow\uparrow 2\uparrow\uparrow 2=2\uparrow\uparrow 2\uparrow 2=2\uparrow\uparrow 4=2\uparrow 2\uparrow 2\uparrow 2=\\
&=2\uparrow 2\uparrow 4=2^{16}=65536
\end{align*}
which is for instance the number of characters which can be stored in a 2-byte system on a computer.
\end{example}
\begin{remark}[Trivial towers]
\label{remark:trivial.towers}
The following equalities hold for any $k\geq1$:
\begin{align*}
\forall x\in\N:1\uparrow^k x &= x\uparrow^k 0 = 1\\
\forall x\in\N:x\uparrow^k 1 &= x\\
2\uparrow^k 2&=4
\end{align*}
\end{remark}
\subsection{Steinhaus-Moser notation}
See \cite{SM} for the original definition.
\begin{defin}
\label{defin:SM}
	Steinhaus-Moser notation uses geometrical shapes to express big numbers. A number surrounded by a shape will have the following meaning:\par\medskip
	\includegraphics[width=5cm]{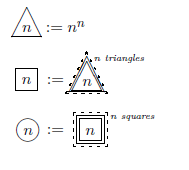}\par\medskip
	Using a more functional notation, we will define ($f^n$ means we compose $f$ with itself $n$ times):
	\begin{itemize}
		\item $\triangle(n):=n^n$;
		\item $\square(n):=\triangle^n(n)$;
		\item $\bigcirc(n):=\square^n(n)$;
	\end{itemize}
	One could also use a regular pentagon instead of the circle and continue the sequence for any regular $k$-agon; we will denote the generalized Steinhaus-Moser notation using the recursive function:
	\begin{align*}
	SM_3(n)&:=n^n=\triangle(n)\\
	SM_{k+1}(n)&:=SM_k^n(n)
	\end{align*}
\end{defin}
\begin{example}
	The number \textbf{Mega} is defined as $\MEGA$, that is:
	$$\bigcirc(2)=\square(\square(2))=\square(\triangle(\triangle(2)))=\square(\triangle(4))=\square(256)=\triangle^{256}(256)$$
	where the last expression contains already too many triangles to be computed explicitly.
\end{example}
\begin{example}
Another important number expressed with this notation is the \textbf{Megiston}, defined as $\MEGISTON$.
\end{example}
\subsection{Meta-algorithms}
All operations we have considered give an ``algorithm'' to compute a natural number; we may construct a ``meta-algorithm'' by considering a string where the instances of ``$\overbrace{\ldots}^{k}$'' mean we should repeat the dotted part $k$ times; for instance:
$$\overbrace{2\uparrow}^{3}5$$
means to construct the algorithm $2\uparrow 2\uparrow 2\uparrow 5$, that is $2^{2^{2^5}}$.\par\smallskip
We write the meta-function ``EXPAND'' meaning the bracketed string should be expanded with the rule just mentioned. We can now define a ``generalized arrow function'' as:
$$\uparrow(A,B,k,C):=\text{EXPAND}\left[\overbrace{A\uparrow^{k-1}}^{B}C\right]$$
so for instance we have the previous ``generalized tetration'': $\uparrow(2,3,2,5)=2^{2^{2^5}}$.\par\smallskip
In general $\uparrow(A,B,k)=\uparrow(A,B,k,1)$, so it is indeed a generalization of the previous definition.
\section{Rooted tree representation}
\subsection{Binary case}
We consider the set $T$ containing the following elements:
\begin{itemize}
	\item $\emptyset\in T$
	\item A finite set of elements of $T$ ($A=\{a_i\in T\}_{i\in I}$) is itself an element of $T$ ($A\in T$), and vice-versa any element of $T$ contains only elements of $T$ without infinite descending chains.
\end{itemize}
This set has the following properties:
\begin{itemize}
	\item Any element $t\in T$ can be associated to a rooted tree: one recursively builds the tree for each element of $t$, and then connects their roots to a new root for $t$ itself. This tree is also \emph{unredundant}, in the sense that different branches of the same node are distinct (from the fact that elements in a set are all different from each other). By this definition, the tree associated to the emptyset will be a root with no branches.
	\item It is defined a \emph{height} function $H:T\rightarrow\N$ as:
	$$H(\emptyset):=0; H(A):=1+\max_{t\in A}{H(t)}$$
	which is well defined from the assumption on descending chains.
	\item There is an ``algorithmic'' bijection $f:T\stackrel{\cong}{\to}\N$ defined recursively as follows:
	\begin{itemize}
		\item $f(\emptyset)=0$;
		\item $f(A)=\sum_{t\in A}2^{f(t)}$.
	\end{itemize}
\end{itemize}
Before going further we briefly prove bijectivity. Indeed, we must prove that $f(A)=f(B)\rightarrow A=B$, and we will proceed by induction on $\max(H(A),H(B))$. We suppose inductively that $f(a)=f(b)\rightarrow a=b$ is true for $\max(H(a),H(b))<\max(H(A),H(B))$. By the uniqueness of the binary expansion for natural numbers, $f(A)$ and $f(B)$ have the same non-zero digits, which correspond to elements $a\in A$, $b\in B$ where $f(a)$ and $f(b)$ give the position of the digit. For each such couple we must have $f(a)=f(b)$ and by the inductive assumption we deduce $a=b$, so that $A$ and $B$ must have the same elements QED.\par\medskip
Using this bijection we are authorized from now on to not distinguish between $A$ and $f(A)$. We define:
\begin{align*}
M_k&:=\max\{A\in T|H(A)=k\}\\
m_k&:=\min\{A\in T|H(A)=k\}
\end{align*}
The first one is obtained when $A$ contains all possible elements $t$ of height $<k$. Thus:
\begin{itemize}
	\item $M_0=0$;
	\item $M_1=2^0=2^1-1=1$;
	\item $M_2=2^0+2^1=2^2-1=3$;
	\item $M_3=2^0+2^1+2^2+2^3=2^4-1=15$;
	\item $M_4=2^0+\ldots+2^{15}=2^{16}-1$.
\end{itemize}
The second one is instead obtained by the recursion $m_0=0$; $m_k=\{m_{k-1}\}$.\par\medskip
Considering the recursive sequence:
\begin{itemize}
	\item $a_0=0$;
	\item $a_{i+1}=2^{a_i}$.
\end{itemize}
one can immediately prove by induction that $M_k=a_{k+1}-1$ and $m_k=a_k$, so that height is proven to be a non-decreasing function. Using Knuth's up-arrow notation, we have $m_k=2\uparrow\uparrow(k-1)$ so that every element of $T$ is found in a specific interval depending on its height:
$$2\uparrow\uparrow(H(A)-1)\leq A<2\uparrow\uparrow H(A)$$
\begin{example}
$M_3=\{\emptyset,\{\emptyset\},\{\{\emptyset\}\},\{\emptyset,\{\emptyset\}\}\}$ is the set of all elements with height at most $2$ so is the greatest one with height $3$, and indeed it satisfies:
$$f(M_3)=2^0+2^1+2^2+2^3=1+2+4+8=15.$$
The associated rooted tree is the following (we write on each node the integer corresponding to its branch):\\
\includegraphics[width=6cm]{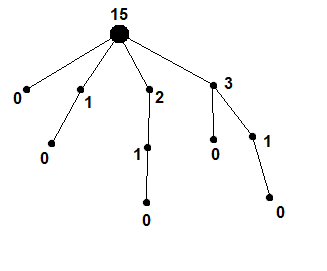}
\end{example}
\begin{remark}
With the usual notation $\mathcal{P}(A):=\{X\subseteq A\}$, we notice that for any $k\geq1$ hold the following facts:
\begin{itemize}
	\item $M_k=\mathcal{P}\left(M_{k-1}\right)$;
	\item $\#(M_k)=m_k$, because $M_k$ contains all numbers from $0$ to $m_k-1$.
\end{itemize}
Summing up those results, we have that the tetration $2\uparrow\uparrow(k-1)=m_k$ represents exactly the cardinality of the set:
$$\mathcal{P}^k(\emptyset):=\text{EXPAND}\left[\overbrace{\mathcal{P}\left(\right.}^{k}\emptyset\overbrace{\left.\right)}^{k}\right].$$
More generally, the ``generalized tetration'' gives the cardinality of the nested power set:
$$\#\mathcal{P}^k(A)=\uparrow(2,k-1,\#(A),2).$$
\end{remark}
\subsubsection{Comparison}
Comparing two elements $A,B\in T$ is performed with the following rule: one recursively can compare elements of $A\Delta B$ (symmetric difference), and put them in order; if its biggest element comes from $A$, then $A$ is the bigger number, otherwise $B$ is the bigger one.
\begin{remark}
For this purpose, and other following purposes, we remind that (as we are talking about sets) the order \emph{in theory} doesn't matter, but actually we should consider every set as being already ordered so that finding the biggest element becomes an easy task.
\end{remark}
\subsubsection{Successor}
We want to compute $s(A)$ for some $A\in T$. If $A=M_k$ for some $k$ then one has to consider directly $s(A)=m_{k+1}$. Otherwise, let $n_A\neq A$ be the unique natural number such that:
$$n_A\not\in A \wedge\forall h<n_A:h\in A$$
which is distinct from $A$ precisely because $A$ is not an $M_k$. Then one just has to remove every $h$ smaller than $n_A$ from $A$ and insert instead the element $n_A$.
\subsubsection{Addition}
The sum of $A$ and $B$ is obtained by joining their elements; if an element $t$ is repeated twice, one performs a \emph{carry} and inserts \textbf{instead} the element $s(t)$, which could as well require another carry.
\subsubsection{Multiplication}
To multiply $A$ and $B$ one considers:
$$A\cdot B=\sum_{(a,b)\in A\times B}(c:=\{a+b\})$$
which in usual representation would mean:
$$\left(\sum 2^a \right)\cdot\left(\sum 2^b \right)=\sum 2^{a+b}.$$
\subsubsection{Generalized rooted tree representation}
General case has been considered first by Goodstein (see \cite{G2}). When considering non-binary bases $b$ the set representation fails to be useful, unless considering a more sophisticated notation. A representation will be a couple $(b,s)$ with $b\in\Z_{\geq2}$ and $s$ a string in the language $\{1,2,\ldots,b-1,``+'',``('',``)''\}$. The string will be interpreted as between digits and brackets there were the full expression ``$\cdot b\uparrow$''. For instance:
$$2(1()+2(1())+1(2())):=2\cdot3\uparrow(1+2\cdot3\uparrow(1)+1\cdot3\uparrow(2))=2\cdot3^{16}.$$
More formally, after fixing the base $b$, one considers the following type of strings:
\begin{itemize}
	\item EMPTY: an empty string \textbf{representing} $0$;
	\item SUM: any number of DIGIT strings (see below) separated by the usual ``+'' symbol and having different exponents, \textbf{representing} the sum of \textbf{values} of the DIGIT components;
	\item MISC: an EMPTY or SUM string;
	\item DIGIT: a digit $0\leq d<b$ followed by a MISC string representing some number $s$ (called ``exponent'') into brackets, which has \textbf{value} $d\cdot b^s$.
\end{itemize}
The final string $m$ has the MISC form, and is associated to a uniquely determined value in $\N$ (precisely the number represented by $m$). This kind of approach is typical of computer science definitions for metadata (see for example \cite{dA}).
\begin{remark}
We recall that again \textbf{order doesn't matter} in SUM strings, as that's the reason we keep using plus symbol as a separator, but for computational purposes one should always consider sums ordered by digits' exponents.
\end{remark}
We also may consider again rooted trees, where now connections between nodes are \textbf{labeled} with a digit from $1$ to $b-1$.
\begin{example}
Using as ``labels'' the colors \textcolor{blue}{\textbf{blue}}=$1$ and \textcolor{red}{\textbf{red}}=$2$, we have the following representation:\par\medskip
\includegraphics[width=3cm]{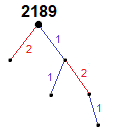}\par\medskip
where the bracketed algorithm is:
$$\left(3,2()+1(1()+2(1()))\right)=2\cdot3^0+1\cdot3^{1\cdot3^0+2\cdot3^{1\cdot3^0}}=2+3^{1+6}=2+2187=2189.$$
We notice that also in this case we can define the \textbf{height} of a graph, and that the sequences of minimum/maximum elements with a certain height can be found as well:
$$b\uparrow\uparrow(H(A)-1)\leq A<b\uparrow\uparrow H(A)$$
because the minimum $m_k=b\uparrow\uparrow(H(A)-1)$ is attained when there is a single path of digits $1$ while the maximum is the sum of terms $(b-1)\times b^k$ with $k<b\uparrow\uparrow(H(A)-1)$, a geometric progression having indeed sum $M_k:=[b\uparrow\uparrow H(A)]-1$.
\end{example}
\subsubsection{Goodstein's theorem}
Goodstein's theorem (see \cite{G2}) has an interesting interpretation within the topic of rooted tree notation. We recall that Goodstein's theorem involves the function which, given a couple $(b,A)$ of a base $b\in\N$ and a rooted tree in that base, can be interpreted as:
$$F(b,A)=(b+1,A-1)$$
where the tree $A$ is reread in the new base $b+1$ and then decreased by 1.\par\medskip
Goodstein's theorem says that iterating this function one definitely stops at the value $0$ whatever is the first element to which it is applied, and even though the function increases dramatically for almost every element. The proof relies on substituting every basis with the ordinal $\omega$, so that the values obtained by this iteration form a strictly decreasing succession of ordinals for which we know it must stop somewhere, and the only possibility is $0$. The rooted tree representation makes clear why the function is decreasing, as any natural number involved in representation is less than $\omega$ in the theory of ordinals.\par\medskip
We also point out that reinterpreting the proof using rooted trees doesn't actually require ordinal theory: geometrical properties of rooted trees should be enough to prove the assert without even involving the base, and this could indeed be studied in a more detailed future work on the topic.\par\medskip
We conclude this section by calculating an effective bound for some Goodstein sequences:
\begin{thm}
\label{thm:goodstein.bound}
Let $b>1$ and $\check{b}:=b-1$. We denote by $B_k(b)$ ($k<b$) the number of steps required for the couple $(\check{b}(\check{k})+\ldots+\check{b}(1)+\check{b}(),b)$ to reach the stopping value $-1$. Then we have an explicit recursion to describe this function:
\begin{align*}
B_1(b)&=2\cdot b\\
B_k(b)&=B_{k-1}^b(b)
\end{align*}
where the latter exponent means one should repeatedly apply $b$ times the function $B_{k-1}$. For example:
$$B_2(b)=\text{EXPAND}\left[\overbrace{2\cdot}^{b}b\right]=2^b b.$$
\end{thm}
\begin{coroll}[\ref{thm:goodstein.bound}]
If $A$ is a tree in the base $b>2$ with height $H(A)\leq2$, then Goodstein's algorithm applied to the couple $(A,b)$ reaches the stopping point $(-1,B)$ when:
$$B=B_b(b)<SM_{b+1}(b)\leq (b+1)\uparrow^{b-1}(b+1)$$
where $SM_k$ is the generalized $k$-agon Steinhaus-Moser function (see definition \ref{defin:SM}) and the last inequality comes from the corollary \ref{prop:general.SM} proved below.\par\medskip
We remark that this corollary tells us that $B-b-1$ is an effective bound for the algorithm to reach $0$.
\end{coroll}
\begin{proof}[\ref{thm:goodstein.bound}]
The first equality comes from the fact that every step decreases the only digit by 1 while increasing the basis by the same amount; thus going from the digit $b-1$ to $-1$ requires $b$ steps, which increase the basis from $b$ to $2b$. The second one derives from the fact that every time the biggest digit decreases by $1$, the other $k-1$ digits come from the same problem where the basis is updated by applying the function $B_{k-1}$, and this has to be done $b$ times.\par\medskip
To prove the corollary, it is known that it is enough to do it for $A=\check{b}(\check{b})+\ldots+\check{b}(1)+\check{b}()=m_2$, and we notice that in this case $B_2(b)=2^b b<3^b\leq b^b=\triangle(b)$ so that the recursive definition forces $B_k(b)<SM_{k+1}(b)$ (compare definition \ref{defin:SM}) and $B=B_b(b)<SM_{b+1}(b)$ as wanted.
\end{proof}
\section{Comparing big numbers}
\subsection{Continued Fractions preliminaries}
\begin{lemma}[Dirichlet property]
\label{lemma:dirichlet.condition}
A continued fraction approximant $\frac ba$ to an irrational number $x>1$ satisfies:
$$\left|\frac ba-x\right|<\frac1{ab}<\frac1{a^2}$$
\end{lemma}
\begin{proof}[\ref{lemma:dirichlet.condition}]
It is well known that $x$ is between $\frac ba$ and the next approximant $\frac cb$, and that:
$$\left|\frac ba-\frac cb\right|=\frac1{ab}<\frac1{a^2}$$
so that the assertion follows immediately.
\end{proof}
\begin{lemma}
\label{lemma:power.comparation}
Given $A<B\in\N$ such that $x=\frac{\ln B}{\ln A}$ is an irrational real number, the continued fraction approximants $\frac ba$ to $x$ are such that:
$$e^{-\varepsilon}<A^b/B^a<e^\varepsilon$$
where $\varepsilon:=\frac{\ln A}{b}$.
\end{lemma}
\begin{proof}[\ref{lemma:power.comparation}]
By lemma \ref{lemma:dirichlet.condition} we have:
\begin{align*}
-\frac1{a^2}&<\frac ba-\frac{\ln B}{\ln A}<\frac1{a^2}\\
-\frac1b&<b-\frac{\ln B}{\ln A}a<\frac1b\\
A^{-\frac1b}&<\frac{A^b}{B^a}<A^{\frac1b}
\end{align*}
and we conclude observing that $e^{\pm\varepsilon}=A^{\pm\frac1b}$ by definition of $\varepsilon$.
\end{proof}
\subsection{Undistinguishable numbers}
See also the introduction to \cite{AL}.
\begin{thm}
\label{thm:undistinguishable}
If $A,B,a,b,x$ are as in lemma \ref{lemma:dirichlet.condition} and $k>1$ is a natural number, then $A^b$ and $B^a$ are $k$- or $(k+1)$-\textbf{undistinguishable powers} when:
$$b>\ln A\cdot 2\cdot 10^{k+1}$$
in the sense that in scientific notation they have the same expression considering only the first $k$ or $k+1$ significant digits of their decimal expansion.
\end{thm}
\begin{proof}[\ref{thm:undistinguishable}]
Two number whose ratio is bounded by the number $\frac{1}{1-0.5\cdot10^{k+1}}\approx 1+0.5\cdot 10^{-(k+1)}\approx\exp(0.5\cdot 10^{-(k+1)})$ are sure to have the same scientific notation expression to the $(k+1)$-th significant digit, possibly differing for the last one (including the possibility of a carry); in this case the $(k+1)$-th digit must be the same (the difference between the two approximations is bigger than double the difference of the two numbers) and we will have the same approximation to the $k$-th digit instead.\par\medskip
Now we can apply lemma \ref{lemma:power.comparation}, where by hypothesis $\varepsilon<0.5\cdot 10^{-(k+1)}$ so that the ratio $A^b$ and $B^a$ is bounded by $e^\varepsilon$, i.e. the number we just talked about, and we know already that in this case the thesis holds.
\end{proof}
\begin{example}
For $A=2$ and $B=3$ we can consider the approximant:
$$\frac ba=[1; 1, 1, 2, 2, 3, 1, 5, 2, 23, 2, 2, 1, 1, 55]=\frac{16785921}{10590737}.$$
Being $b>\ln 2\cdot 2\cdot 10^7\approx 13862944$, we know that $2^{16785921}$ and $3^{10590737}$ are $6$-undistinguishable powers, and indeed both have the following expression in scientific notation:
\begin{align*}
5.3191952\ldots\cdot10^{5053065}&\approx 5.31920\cdot10^{5053065}\\
5.3191955\ldots\cdot10^{5053065}&\approx 5.31920\cdot10^{5053065}
\end{align*}
that is, they give the same approximation to the $6$-th digit (one of them actually approximate to $5.319196$ to the $7$-th digit, so we must take one digit less for the exact correspondence).
\end{example}
\subsection{Comparing Knuth and Steinhaus-Moser notations}
We will consider only positive integers when not specified otherwise. Moreover $k$ will be a counter ranging from $0$ to $n$.
\begin{prop}
\label{prop:square}
	The square symbol is comparable to the tetration in the following way:
	$$n\uparrow\uparrow (n+1) \leq\square(n)\leq n\uparrow n\uparrow (n+1)\uparrow\uparrow (n-1)\leq n\uparrow\uparrow(n+2).$$
	More precisely, a sequence of $k$ triangles has the property:
	$$n\uparrow\uparrow (k+1) \leq\triangle^k(n)\leq n\uparrow n\uparrow (n+1)\uparrow\uparrow (k-1).$$
\end{prop}
\begin{proof}[\ref{prop:square}]
 The first inequality is straightforward, as we have by induction:
 $$\triangle(n\uparrow\uparrow k)={[n\uparrow\uparrow k]}^{n\uparrow\uparrow k}=n^{n^{\left\{[n\uparrow\uparrow (k-2)] +[n\uparrow\uparrow (k-1)]\right\}}}\geq n^{n^{\left[n\uparrow\uparrow (k-1)\right]}}=n\uparrow\uparrow (k+1)$$
	so that $n$ in $k$ triangles is always $\geq n\uparrow\uparrow (k+1)$.\par\medskip
	For the second inequality, we notice that:
	$$\triangle\left(n^{n^{[(n+1)\uparrow\uparrow (k-2)]}}\right)=\left\{n^{n^{[(n+1)\uparrow\uparrow (k-2)]}}\right\}^{n^{n^{[(n+1)\uparrow\uparrow (k-2)]}}}=n^{n^{\left\{[(n+1)\uparrow\uparrow (k-2)] +[n\uparrow (n+1)\uparrow\uparrow (k-2)\right\}}}$$
	and:
	$$[(n+1)\uparrow\uparrow (k-2)] +[n\uparrow(n+1)\uparrow\uparrow (k-2)]\leq (n+1)\uparrow\uparrow (k-1)$$
	as developing $(n+1)\uparrow(n+1)\uparrow\uparrow(k-1)$ with Pascal's triangle one obtains, among the others, the term $[(n+1)\uparrow\uparrow (k-2)]n^1 1^{(n+1)\uparrow\uparrow (k-2)-1}$ which is greater than the first term of the addition.\par\medskip
	Thus, $n$ in $k$ triangles is always $\leq n^{n^{[(n+1)\uparrow\uparrow (k-1)]}}$.\par\medskip
	Both inductions start from the case $k=1$, for which all three quantities are trivially equal to $n^n$ (using the rules from remark \ref{remark:trivial.towers}).
\end{proof}
\begin{lemma}
\label{lemma:tetration.of.tetrations}
$$(A\uparrow\uparrow B)\uparrow\uparrow C\leq A\uparrow\uparrow(B+C)$$
\end{lemma}
\begin{proof}[\ref{lemma:tetration.of.tetrations}]
We start by excluding the trivial cases $A=1\vee B=1$. We will use the abbreviation $E:=A\uparrow\uparrow(B-1)$.\par\medskip
We prove more specifically that:
$$(A\uparrow\uparrow B)\uparrow\uparrow C\leq A\uparrow A\uparrow (A+1)\uparrow\uparrow (B+C-3)$$
The original estimate is then a tower one level higher but replacing all $A+1$ with $A$, thus abundantly bigger. We proceed by induction, after checking that the case $C=1$ is trivial. For the induction step we see immediately that:
\begin{align*}
(A\uparrow\uparrow B)\uparrow\uparrow (C+1)&=A\uparrow[E \times (A\uparrow\uparrow B)\uparrow\uparrow C]\leq\\
&\leq A\uparrow[E \times A\uparrow A\uparrow (A+1)\uparrow\uparrow (B+C-3)]=\\
&=A\uparrow A\uparrow[A\uparrow\uparrow(B-2)+A\uparrow (A+1)\uparrow\uparrow (B+C-3)]
\end{align*}
so the thesis follows from the following elementary inequality:
\begin{align*}
A\uparrow\uparrow (B-2)+ A\uparrow (A+1)\uparrow\uparrow (B+C-3)&\leq (A+1)\uparrow (A+1)\uparrow\uparrow (B+C-3)=\\
&=(A+1)\uparrow\uparrow (B+C-2).
\end{align*}
\end{proof}
\begin{prop}
\label{prop:circle}
	The circle symbol (see next section for the case of \textbf{Mega}) is comparable to the pentation in the following way:
	$$n\uparrow\uparrow\uparrow (n+1) \leq\bigcirc(n)\leq n\uparrow\uparrow(n+1)\uparrow\uparrow\uparrow n.$$
	More precisely, a sequence of $k$ squares has the property:
	$$n\uparrow\uparrow\uparrow (k+1)\leq\square^k(n)\leq n\uparrow\uparrow(n+1)\uparrow\uparrow\uparrow k.$$
\end{prop}
\begin{proof}[\ref{prop:circle}]
	As before, the first inequality is straightforward (using induction) by proposition \ref{prop:square}:
	$$\square\left(n\uparrow\uparrow\uparrow k\right)\geq (n\uparrow\uparrow\uparrow k)\uparrow\uparrow(1+n\uparrow\uparrow\uparrow k)\geq n\uparrow\uparrow(n\uparrow\uparrow\uparrow k)=n\uparrow\uparrow\uparrow (k+1)$$
	The second inequality can be proved by induction using proposition \ref{prop:square} and lemma \ref{lemma:tetration.of.tetrations}:
	\begin{align*}
		\square\left(n\uparrow\uparrow(n+1)\uparrow\uparrow\uparrow k\right)&\leq [n\uparrow\uparrow(n+1)\uparrow\uparrow\uparrow k]\uparrow\uparrow[2+n\uparrow\uparrow(n+1)\uparrow\uparrow\uparrow k] \leq\\
		&\leq n\uparrow\uparrow[2 + (n+1)\uparrow\uparrow\uparrow k + n\uparrow\uparrow(n+1)\uparrow\uparrow\uparrow k] \leq\\
		&\leq n\uparrow\uparrow(n+1)\uparrow\uparrow\uparrow (k+1).
	\end{align*}
	We point out that both inequalities when $k=0$ become equalities (using the rules from remark \ref{remark:trivial.towers}).
\end{proof}
\begin{lemma}
\label{lemma:composing.arrows}
When $k\geq2$ one has:
$$(A\uparrow^k B)\uparrow^k C\leq A\uparrow^k(B+C)$$
\end{lemma}
\begin{proof}[\ref{lemma:composing.arrows}]
We start by excluding the trivial cases $A=1\vee B=1$. We proceed by induction on $k$, remarking that lemma \ref{lemma:tetration.of.tetrations} gives the starting case $k=2$, thus supposing that the assertion holds already for $k-1$. We will use the abbreviation $E:=A\uparrow^k(B-1)$.\par\medskip
We prove more specifically that:
$$(A\uparrow^k B)\uparrow^k C\leq A\uparrow^{k-1} (A+1)\uparrow^k (B+C-2)$$
The original estimate is then a tower of $\uparrow^{k-1}$-hyperoperations one level higher but replacing all $A+1$ with $A$, thus abundantly bigger. We now proceed by induction on $C$, after checking that the case $C=1$ is trivial. For the induction step we see immediately that:
\begin{align*}
(A\uparrow^k B)\uparrow^k (C+1)&=(A\uparrow^{k-1} E)\uparrow^{k-1}[(A\uparrow^k B)\uparrow^k C]\\
&\leq A\uparrow^{k-1}[E + (A\uparrow^k B)\uparrow^k C]\leq\\
&\leq A\uparrow^{k-1} [E + A\uparrow^{k-1} (A+1)\uparrow^k (B+C-2)]\leq\\
&\leq A\uparrow^{k-1} [(A+1)\uparrow^k (B+C-1)]
\end{align*}
as expected.
\end{proof}
\begin{prop}
\label{prop:general.SM}
	The Steinhaus-Moser generalized function is comparable to Knuth's up-arrow notation in the following way:
	\begin{align*}
	n\uparrow^m(n+1) \leq SM_{m+2}(n)&\leq n\uparrow^{m-1}(n+1)\uparrow^m n<\\
	&<(n+1)\uparrow^m (n+1).
	\end{align*}
	More precisely:
	$$n\uparrow^m(k+1)\leq SM_{m+1}^k(n)\leq n\uparrow^{m-1}(n+1)\uparrow^m k.$$
\end{prop}
\begin{proof}[\ref{prop:circle}]
	We point out again that both inequalities when $k=0$ become trivial equalities (using the rules from remark \ref{remark:trivial.towers}). As before, with a straightforward double induction ($m$/$k$) we can prove the first inequality:
	\begin{align*}
	SM_{m+1}^k(n)&=SM_{m+1}(SM_{m+1}^{k-1}(n))\geq\\
	&\geq SM_{m+1}^{k-1}(n)\uparrow^{m-1}(SM_{m+1}^{k-1}(n)+1)\geq\\
	&\geq\left[n\uparrow^m k\right]\uparrow^{m-1}\left[n\uparrow^m k\right]\geq\\
	&\geq n\uparrow^{m-1}\left[n\uparrow^m k\right]=n\uparrow^m(k+1)
	\end{align*}
	The second inequality can be proved, again by double induction, with lemma \ref{lemma:composing.arrows}:
	\begin{align*}
		SM_{m+1}(SM_{m+1}^k(n))&\leq [1+SM_{m+1}^k(n)]\uparrow^{m-1}[1+SM_{m+1}^k(n)]\leq\\
		&\leq n\uparrow^{m-1}[1+SM_{m+1}^k(n)+(n+1)\uparrow^m k]\leq\\
		&\leq n\uparrow^{m-1}(n+1)\uparrow^m (k+1).
	\end{align*}
\end{proof}

\subsubsection{Examples: Mega and Megiston}
Proposition \ref{prop:square} lets us have bounds for the number \textbf{Mega} as follows:
$$256 \uparrow\uparrow 257 \leq\MEGA\leq 256 \uparrow 256 \uparrow 257 \uparrow\uparrow 255\leq 257\uparrow\uparrow257$$
because as we have seen before it can also be expressed as $\square(256)$.\par\medskip
The \textbf{Megiston} is instead approximable at pentation level by proposition \ref{prop:circle} with the following bounds:
$$10\uparrow\uparrow\uparrow 11 \leq\MEGISTON\leq 10\uparrow\uparrow11\uparrow\uparrow\uparrow10\leq11\uparrow\uparrow\uparrow11.$$
\section*{Acknowledgments}
This work has been partially supported by POR Calabria FESR-FSE 2014--2020, with the grant for research project ``IoT\&B'', CUP J48C17000230006.

\end{document}